\documentclass[a4paper]{article}
\usepackage{graphicx}
\usepackage{cite}
\usepackage{algorithm}
\usepackage{amsmath, amsfonts}
\usepackage{algorithmic}
\usepackage{color}
\usepackage{textcomp}
\usepackage{amsbsy}
\usepackage{latexsym}

\newtheorem{theorem}{Theorem}
\newtheorem{ass}{Assumption}%[theorem]{Assumption}

\newtheorem{lemma}{Lemma}
\newtheorem{proof}{Proof}
\newtheorem{prob}{Problem}

\newcommand{\A}{{\cal A}}
\newcommand{\B}{{\cal B}}

\newcommand{\C}{{\cal C}}
\newcommand{\D}{{\cal D}}

\newcommand{\Hm}{{\mathbf H}}

\begin{document}

\title{System Identification Under Multi-rate Sensing Environment}
\author{Hiroshi Okajima, Risa Furukawa and Nobutomo Matsunaga}
%\address{Kumamoto University\\
%	Kumamoto City Chuo-Ku Kurokami 2-39-1\\
%	E-mail: okajima@cs.kumamoto-u.ac.jp}
\markboth{H. Okajima, R. Furukawa and N. Matsunaga}{System Identification Under Multi-rate Sensing Environment}
%\dates{2000/00/00}{2000/00/00}
\maketitle

\begin{abstract}% Do not delete this percent symbol
	\noindent (Abstract) This paper proposes a system identification algorithm for systems with multirate sensors in a discrete-time framework. It is challenging to obtain an accurate mathematical model when the ratios of the inputs and outputs are different. A cyclic reformulation-based model for multirate systems is formulated, and the multirate system can be reduced to a linear time-invariant system to derive the model under the multirate sensing environment. The proposed algorithm integrates a cyclic reformulation with a state coordinate transformation of the cycled system to enable precise identification of systems under the multirate sensing environment. The effectiveness of the proposed system identification method is demonstrated through numerical simulations. 
\end{abstract}

\section{Introduction}\label{sec1}

System identification has been a well-established research field for decades and continues to evolve with ongoing advancements\cite{id0,id1,id2,id3,id4,id5}. It is a fundamental process for constructing mathematical models of unknown dynamic systems, which play a crucial role in enhancing the accuracy and performance of control system design. In applications such as optimal control for robotics and process control, accurately modeling system dynamics is essential for achieving the desired performance. A representative control methodology that utilizes mathematical models is model-based control, where the mathematical model is explicitly integrated into the controller. Recent research has focused on extending system identification techniques to address challenges in nonlinear systems \cite{non1}, periodic systems \cite{lpv, lptv}, and their practical applications in real-world systems \cite{ouyou}.

This paper focuses on control systems that incorporate multiple sensors. In many control systems, multiple sensors are employed to construct a feedback control loop. When these sensors are of different types, their sampling rates also differ. In such cases, system identification must be performed using signals with different sampling rates for inputs and outputs.

Several approaches have been proposed for multirate systems. In \cite{rev1-5}, an optimization-based algorithm for state estimation of multirate control systems was presented to ensure convergence in state estimation between sampling instants. For control applications, in \cite{rev1-2}, an adaptive servo control law with input constraints was derived for linear time-invariant systems with unknown parameters and two sampling rates. Similarly, in \cite{rev1-6}, a self-tuning control algorithm was developed for discrete-time, linear, multi-variable systems with skew output sampling and input updating mechanisms. Additionally, in \cite{rev1-e}, a parametric design method for dual-rate controllers was proposed to improve steady-state intersample response.

The demand for multirate system identification is particularly evident in mobile robot control \cite{mobile0,mobile1,mobile2,mobile3,mobile4}, where various sensors with different sampling capabilities are employed for self-localization, perception, and control. When implementing model-based control frameworks for mobile robots, system identification must be performed in these multirate sensor environments, making this an important area of research.

A multirate control system can be interpreted as a system where the input and output datasets have missing values. System identification must be performed for such systems, making the problem significantly more challenging compared to the mathematical modeling of a linear time-invariant system with complete data.

Several methods exist for system identification of multirate systems\cite{senkou01,senkou02}. 
In \cite{senkou01}, a multirate identification problem is addressed by dividing the multirate sampled system into different subsystems, and a multirate distributed model predictive control technique is proposed. In \cite{senkou02}, the multi-innovation identification theory is applied to estimate
the parameters of the multirate input models and to present a multi-innovation stochastic gradient algorithm for the multirate systems from the input and output data. The system in \cite{senkou02} is a system for the multirate input signal. However, the issue of multirate in the output presents a more challenging problem and has not been resolved in this research. Furthermore, the identification of general single-input and single-output multirate sampled-data systems was addressed using lifting techniques to extract fast-rate discrete-time systems in \cite{rev1-1}. Similarly, in \cite{rev1-4}, subspace identification algorithms were employed to identify lifted multirate systems and applied them to long-range predictive control. While these studies made significant contributions, there remains room for further investigation into systems with heterogeneous sampling rates across multiple sensors. In particular, systems identified as time-invariant systems through lifting contain variable products, making it difficult to revert back to the original multirate system. Therefore, solving the system identification problem for systems with multirate sensing environments is both beneficial and important.

This paper proposes a system identification algorithm for systems with multirate sensors in a discrete-time framework. Accurately deriving a mathematical model becomes challenging when the input and output sampling rates differ. In our previous study \cite{okajima}, a cyclic reformulation-based system identification algorithm was proposed for linear periodically time-varying systems (LPTV systems). A similar technique for the system identification for LPTV systems is applied in this paper. First, we formulate the multirate system as a structure of the LPTV system. Then, the formulated multirate system is reduced to a linear time-invariant system using cyclic reformulation, which is proposed by Bittanti et al. (\cite{cyc1,cyc2}). In this paper, we extend the results of the previous study \cite{okajima} to solve the identification problem of multirate control systems. 

This paper is organized as follows. Section \ref{sec2} defines the state-space representation of the multirate system. Then, we describe the representation method of cyclic reformulation, which is a time-invariant method used to handle multirate systems as time-invariant systems in Section \ref{sec3}. We also describe the properties of the cyclic reformulation. In Section \ref{sec4}, we propose a system identification algorithm for the multirate systems. The proposed system identification algorithm integrates a cyclic reformulation with a state coordinate transformation of the cycled system to enable precise identification of systems with multirate sensors. In Section \ref{sec5}, we present the effectiveness of the proposed system identification algorithm using numerical examples and verifying the obtained model.

\section{Problem Formulation}\label{sec2}
This section describes the system representation of multirate systems in the context of discrete-time periodic time-varying systems. As a preparatory step before addressing multirate systems, we first consider single-rate systems (as discrete-time linear time-invariant systems). Specifically, the linear time-invariant system is expressed as shown in (\ref{siki1}).
\begin{eqnarray} 
x(k+1)&=&Ax(k)+Bu(k)\label{siki1}\\
y(k)&=&Cx(k)+ Du(k)\label{siki12}
\end{eqnarray}
$k$ is the sampling period, $x \in R^{n}$ is a state, $u \in R^{m}$ and $y \in R^{l}$ are input and output for the system. $A \in R^{n\times n}, B \in R^{n\times m}, C \in R^{l\times n}, D \in R^{l\times m}$ are given. The system (\ref{siki1}), (\ref{siki12}) is assumed to be controllable and observable. In addition, we assume the matrix rank of $A$ is $n$. 

Next, based on the single-rate system (\ref{siki1}), (\ref{siki12}), we consider the representation of a multirate system, where the output observation periods differ from the control period\cite{multi1,multi2}. We obtain the following state-space system (\ref{eq:P_multi}), (\ref{eq:P_multi2}) as a representation of the multirate system using the period-time-varying matrix $V_k$.
\begin{eqnarray}
%\begin{split}
x(k+1)&=&Ax(k)+Bu(k)\label{eq:P_multi}\\
y(k)&=&V_kCx(k)+V_kDu(k) \label{eq:P_multi2}
%\end{split}
\end{eqnarray}
The observation periods of the outputs, which are the components in $y(k)$, are assumed to be different. Each observation period of the outputs is set as a natural multiple of the control period. $M_1, \cdots, M_l$ are the observation periods of the outputs $y_1(k), \cdots, y_l(k)$. For example, the output $y_1(k)$ is observed once every $M_1$ step, and $y_1(k)=0$ at other times. 

The periodically time-varying matrix $V_k$ is determined as follows to characterize the above-mentioned observation period.
When $M$ is the least common multiple of the observation period $M_1, \cdots, M_l$, $V_k$ is a matrix to characterize the observation period of the output, and its period is given as $M$.
\begin{eqnarray}\label{eq:Sk}
V_{k}=\mathrm{diag}\left[\begin{array}{ccccc}
 v_{1}(k) & \cdots & v_{i}(k) & \cdots & v_{l}(k) 
\end{array}\right]
\end{eqnarray}
The element $v_i(k)$ in the matrix corresponds to $y_i(k)$ with the $i$-th observation period $M_{i}$. The component $v_i(k)$ is set to $1$ when the output signal $y_i$ is observed because the observation period $M_i$ is a divisor of $M$ for any $i$, and the timing at which the output signal is not observed is set to $v_i(k)=0$. Let $V_k$ be a matrix of period $M$, and the $i$-th element $v_i(k)$ equals $1$ exactly $M/M_i$ times in an $M$ period. 

Since the least common multiple of $M_1, \cdots, M_m$ is $M$, the sensing timing of all sensors can be represented using a framework with period $M$. Therefore, the following equation holds.
\begin{eqnarray}\label{eq:Sk2}
V_{k}=V_{k\,\bmod\,M}
\end{eqnarray}
In this case, $\{V_0, \cdots, V_{M-1}\}$ can be prepared to represent a multirate system. 

Then, we give an assumption about the observability of the multirate system, which is expressed as $M$-periodic system, as stated in Assumption \ref{ass2}. 
\begin{ass}\label{ass2}
There exists at least one $j$ that an observability pair $(V_j C, A^M)$ is observable. In other words, an observability matrix rank is $n$ for the observability pair $(V_j C, A^M)$. 
\end{ass}
From the above, the multirate system (\ref{eq:P_multi}), (\ref{eq:P_multi2}) is represented as a periodic time-varying system with period $M$. By adopting this representation, the characteristics of the multirate system are incorporated into $V_k$, reducing the parameters to be identified in modeling to four parameters: $A, B, C, D$ (Figure. \ref{figure1}). 

\begin{prob}\label{prob1}
Consider the case where input data $\{u(k)\}$ is applied to the system under a multirate sensing environment (\ref{eq:P_multi}), (\ref{eq:P_multi2}) to obtain output data $\{y(k)\}$. Based on the input-output data $\{u(k), y(k)\}_{k=0}^{N-1}$, estimate $A$, $B$, $C$ and $D$ up to state coordinate transformation.
\hfill $\Box$
\end{prob}
We denote the system matrices $A_m$, $B_m$, $C_m$, and $D_m$ as the solutions of Problem \ref{prob1}. 
%Since we are dealing with an $M$-periodic multirate system, from (\ref{eq:P_multi}), we need only matrices $A_m$, $B_m$, $C_m$, and $D_m$ in Problem \ref{prob1}. 

\begin{figure}[!h]
\centering
\includegraphics[width=0.45\textwidth]{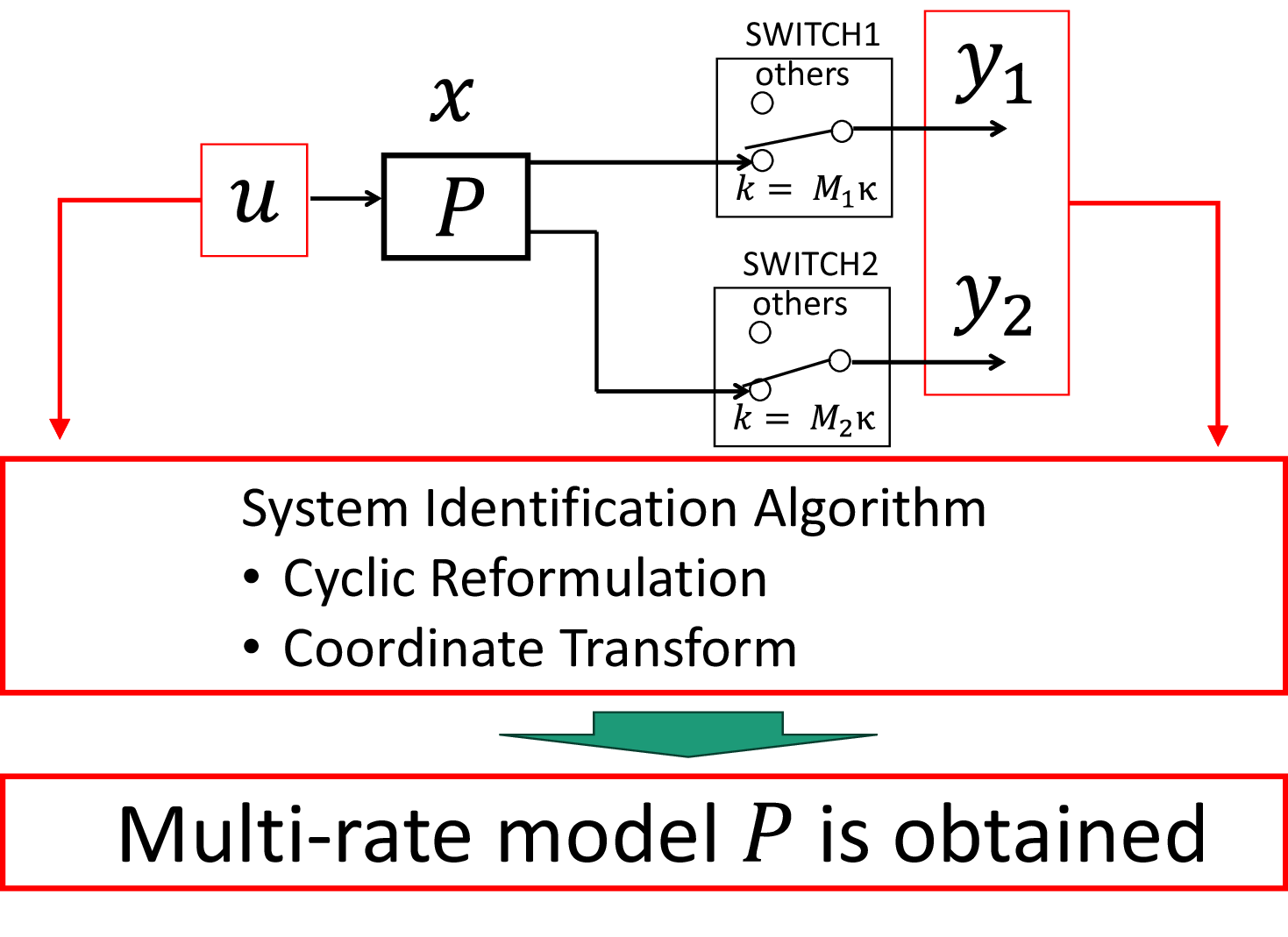}
\caption{System Identification for Multirate Sensing Environments}
\label{figure1}
\end{figure}

In addition to the observability conditions, we consider the following matrix $F \in {R}^{n\times nl}$ for convenience. 
\begin{eqnarray}
F = \begin{bmatrix} F_0 & F_1 & \cdots & F_{n-1}\end{bmatrix}, \label{largef}
\end{eqnarray}
Their elements $F_i (i=0,\cdots,n-1)$ are $n\times l$ matrices. When $(V_j C,A^M)$ is observable, it is possible to choose appropriate matrices $F$ such that the matrix rank of the following $n\times n$ matrices $X_j$ becomes $n$. 
\begin{eqnarray}
X_j := F \begin{bmatrix}V_{j}C\\V_{j}C A^M \\V_{j} C A^{2M} \\ \vdots \\ V_{j} CA^{(n-1)M}\end{bmatrix}\label{fk} 
\end{eqnarray}
%Then, we determine a matrix $\check X$ as follow:
%\begin{eqnarray}
%\check X = \bf{diag}\{X_0,X_1,\cdots,X_{M-1}\}\label{checkx}
%\end{eqnarray}
%The matrix rank of $\check X$ is $Mn$ by appropriate choice of $F$. 

As a simple example of choosing $F$, we consider the case where $l = 1$. The size of each $F_j$ is $n \times 1$. The following components $(F_j)_i$: 
\begin{eqnarray}\label{fji}
(F_j)_i = \begin{cases} 1, i = j\\0, i\neq j\end{cases}
\end{eqnarray}
is one obvious choice for $F_j$ such that $F$ satisfy the rank condition of $X_j$. Then, we can see that $F$ is given as an identity matrix $I_{n}$. It is not difficult to choose matrices $F$ which satisfy the conditions of matrix rank. 

%%%%%%%%%%%%%%%%%%%%%%%%%%%%%%%%%%%%%%%%%%%%%%%
\section{Cyclic Reformulation of Multirate System}\label{sec3}
%%%%%%%%%%%%%%%%%%%%%%%%%%%%%%%%%%%%%%%%%%%%%%%
\subsection{Time-invariant system expression} \label{sec22}
The cyclic reformulation, which is a time-invariant system expression method, is introduced in this section. 
The lifting reformulation is the most traditional method for obtaining a time-invariant system from an LPTV system. The lifting operation consists of packaging the values of a signal over one period in an extended signal. On the other hand, a method used in this paper is a cyclic reformulation \cite{cyc1,cyc2}. 

First, a cycled input signal is determined based on the input for (\ref{eq:P_multi}), (\ref{eq:P_multi2}) as follows. 
\begin{eqnarray}
\lefteqn{\check u(0) = \begin{bmatrix}u(0)\\O_{m,1}\\ \vdots\\ O_{m,1}\end{bmatrix}, \check u(1) = \begin{bmatrix}O_{m,1}\\u(1)\\ \vdots\\ O_{m,1}\end{bmatrix}, \cdots, }\label{checku}\\&\check u(\!M\!-\!1\!) = \begin{bmatrix}O_{m,1}\\ \vdots\\O_{m,1}\\ u(\!M\!-\!1\!)\end{bmatrix}, \check u(M) = \begin{bmatrix}u(M)\\O_{m,1}\\ \vdots\\ O_{m,1}\end{bmatrix} \cdots \nonumber%\\ &% \\&(k = \kappa N) \nonumber
\end{eqnarray}
The cycled input $\check u(k)\in {R}^{Mm}$ is obtained by using the input $u(k)$. $\check u(k)$ has a unique non-zero sub-vector $u(k)$ at each time-point. The sub-vector $u(k)$ cyclically shifts along the column blocks. 

Then, the cyclic reformulation of the $M$-periodic system (\ref{eq:P_multi}), (\ref{eq:P_multi2}) is described by
\begin{equation}
 \label{eqcyclic}
     \begin{array}{rcl}
          \check{x}(k+1) & = & \check{A}\check{x}(k) + \check{B}\check{u}(k)\\
          \check{y}(k) & = & \check{C}\check{x}(k) + \check{D}\check{u}(k), 
     \end{array}
\end{equation}
where matrices $\check{A}, \check{B}, \check{C}, \check{D}$ are given as follows. 
\begin{equation}
     \check{A} = \left[\begin{array}{ccccc}
          O_{n,n} & \cdots & \cdots& O_{n,n} & A\\
          A &  O_{n,n} & \cdots &  O_{n,n} &  O_{n,n}\\
           O_{n,n} & A & \ddots & \vdots & \vdots\\
          \vdots & \ddots & \ddots &  O_{n,n} & \vdots\\
           O_{n,n} & \cdots &  O_{n,n} & A &  O_{n,n} 
     \end{array} \right],\label{checka}
\end{equation}
\begin{equation}
     \check{B} = \left[\begin{array}{ccccc}
           O_{n,m} & \cdots & \cdots& O_{n,m} & B\\
          B & O_{n,m} & \cdots & O_{n,m} & O_{n,m}\\
          O_{n,m} & B & \ddots & \vdots & \vdots\\
          \vdots & \ddots & \ddots & O_{n,m} & \vdots\\
          O_{n,m} & \cdots & O_{n,m} & B & O_{n,m} 
     \end{array} \right],\label{checkb}
\end{equation}
\begin{equation}
     \check{C} = \left[\begin{array}{cccc}
          V_0 C & O_{l,n} & \cdots &O_{l,n} \\
          O_{l,n}  & V_1 C & \ddots & \vdots\\
          \vdots & \ddots & \ddots & O_{l,n}\\
          O_{l,n}  & \cdots & O_{l,n}  & V_{M-1}C
     \end{array} \right],
\end{equation}
\begin{equation}
     \check{D} = \left[\begin{array}{cccc}
          V_0 D & O_{l,m} & \cdots &  O_{l,m} \\
           O_{l,m}  & V_{1}D & \ddots & \vdots\\
          \vdots & \ddots & \ddots &  O_{l,m} \\
           O_{l,m}  & \cdots &  O_{l,m}  & V_{M-1}D
     \end{array} \right].
\end{equation}
The dimensions of each matrix are given as $\check{A}\in {R}^{Mn\times Mn}$, $\check{B}\in {R}^{Mn\times Mm}$, $\check{C}\in {R}^{Ml\times Mn}$ and $\check{D}\in {R}^{Ml\times Mm}$. The matrix structures of $\check{A}$ and $\check{B}$ are named as cyclic matrices. The structures of $\check{C}$ and $\check{D}$ are block diagonal matrices. The dimensions of the state and output are given as $\check{x}(k)\in {R}^{Mn}$, $\check{y}(k)\in {R}^{Ml}$. 

The initial state $\check{x}(0)$ is given by using $x(0)$ as follows.  
\begin{eqnarray}
\check{x}(0) = \left[\begin{array}{c}
     x(0) \\ O_{n,1} \\ \vdots \\ O_{n,1}\end{array}\right].\label{siki18}
\end{eqnarray}
Then, $\check x(1)$ can be obtained by using (\ref{eqcyclic}), (\ref{siki18}) with $\check u(0)$ as follows. 
\begin{eqnarray}
     \check{x}(1) = \left[\begin{array}{c}
           O_{n,1} \\
          Ax(0) + Bu(0)\\
           O_{n,1}\\
          \vdots\\
           O_{n,1}
     \end{array} \right]
\end{eqnarray}
A sub-vector in $\check x(1)$ exactly corresponds to $x(1)$ as defined in (\ref{siki1}). 
Furthermore, we can obtain the cycled state signal $\check x(k)$ and the cycled output signal $\check y(k)$ by using (\ref{eqcyclic}) and cycled input signal $\check u(k)$ by using step by step calculation. 

\subsection{Controllability and observability of cycled system}
The characteristics of observability and controllability of the cycled system (\ref{eqcyclic}) are presented in this section. 

For the cycled system (\ref{eqcyclic}), a controllability matrix can be written as follows. 
\begin{eqnarray}
{\Psi}_c = \left[\check B, \check A\check B, \cdots, \check A^{Mn-1}\check B\right]
\end{eqnarray}
The controllability of the system (\ref{eqcyclic}) is automatically satisfied when the pair $(A,B)$ is controllable. Then, the following condition is satisfied for the cycled system. 
\begin{eqnarray}
{\bf{rank}} {\Psi}_c = Mn
\end{eqnarray}

Then, an observability matrix for the $M$-periodic multirate system can be written as follows. 
\begin{eqnarray}
{\Psi}_o= \left[\begin{array}{c}\check C\\ \check C\check A\\ \vdots\\ \check C\check A^{Mn-1}\end{array}\right]
\end{eqnarray}
The matrix size of ${\Psi}_o$ is $Mln\times Mn$. 
We prove the observability of the pair $(\check C,\check A)$ by using Assumption \ref{ass2}, and the matrix rank of $A$ is $n$. 

First, we assume a pair $(V_j C, A^M)$ is observable based on Assumption \ref{ass2}. By using elementary row transformations for the observability matrix $\Psi_o$, a matrix $\tilde \Psi_o \in R^{Mln\times Mn}$ can be obtained as follows. 
\begin{eqnarray}
%\tilde \Psi_o = \begin{bmatrix},\cdots,\psi_j,\cdots,\psi_{M-1}\right]\end{bmatrix}
     \tilde \Psi_o = \left[\begin{array}{ccccc}
          \psi_0 & O_{n,n} & \cdots& O_{n,n} &  O_{n,n} \\
          O_{n,n}&  \ddots & O_{n,n}  &  \ddots  & O_{n,n} \\
          O_{n,n} & O_{n,n} & \psi_j & \ddots & \vdots\\
          \vdots & \ddots & \ddots & \ddots & O_{n,n} \\
           O_{n,n} & \cdots &  O_{n,n} & O_{n,n} & \psi_{M-1}\\
           * & * & * & * & * \\
            \vdots & \vdots & \vdots & \vdots & \vdots \\
             * & * & * & * & * 
     \end{array} \right]
\end{eqnarray}
Where the elements $\tilde \psi_i (i=0,\cdots,M-1)$ are selected as follows. 
\begin{eqnarray}
%\psi_j &=& \left[\begin{array}{c}V_j C\\ V_j C A^M\\ \vdots\\ V_j C A^{M(n-1)}\end{array}\right]\\
\psi_{i} &=&  \left[\begin{array}{c}V_j C\\ V_j C A^{M}\\ \vdots\\ V_j C A^{M(n-1)}\end{array}\right]A^{i-j+M}, i = 0,\cdots,M-1 \label{psij}
\end{eqnarray}
It is obvious that the matrix rank of $\psi_{i}$ is $n$ for all $i$. Therefore, we can see the matrix rank of $\tilde \Psi_o$ is $Mn$.

Since the matrix rank of $\tilde \Psi_o$ is $Mn$, the matrix rank of $\Psi_o$ is also $Mn$ because $\tilde \Psi_o$ is derived from elementary row transformation of $\Psi_o$. Then, the observability of the pair $(\check C,\check A)$ is satisfied. The following condition holds. 
\begin{eqnarray}
{\bf{rank}} {\Psi}_o = Mn
\end{eqnarray}
Note that the controllability and observability conditions for the LPTV system in \cite{okajima} are more difficult to satisfy than the conditions in this paper.

\subsection{Characteristics of Markov parameters}
The characteristics of Markov parameters for cyclic reformulation of the multirate system are presented. Markov parameters of LPTV systems were characterized in \cite{okajima}. Since the multirate system can be regarded as an LPTV system, the characteristics of Markov parameters are satisfied for the multirate system. 

Markov parameters $\check H(i)$ are known as coefficients of the impulse response of the system (\ref{eqcyclic}) and given by using state-space parameters $\check A$, $\check B$, $\check C$, and $\check D$. 
\begin{eqnarray}
\check H(i) = \begin{cases}\check D,&i = 0 \\ \check C \check A^{i-1} \check B, & i=1,2,\cdots \end{cases}\label{markov}
\end{eqnarray}
For a given positive integer $q$, we introduce a matrix $\check{S}_q$ as follows. 
\begin{equation}
     \check{S}_q = \left[\begin{array}{ccccc}
          O_{q,q} & I_q & O_{q,q} & \cdots & O_{q,q}\\
          O_{q,q} & O_{q,q} & I_q & \ddots & \vdots\\
          \vdots & \ddots & \ddots & \ddots & O_{q,q}\\
          O_{q,q} & \ddots & \ddots & \ddots & I_q\\
          I_q & O_{q,q} & \cdots & \cdots & O_{q,q} 
     \end{array} \right]   \label{sq}
\end{equation}
$\check S_q$ is determined for operating $M$-periodic systems. The matrix size of $\check S_q$ is $Mq\times Mq$. $\check S_q$ is a regular matrix, and its inverse matrix is a cyclic matrix. For any given block diagonal matrix $E \in R^{Mq\times Mq}$ with $q\times q$ block elements $E_i$, $\check S_q^{-1}E \check S_q$ also becomes a block diagonal matrix. It should be noted that the individual block elements $E_i$ in $\check S_q^{-1}E \check S_q$ are shifted by one element relative to $E$. 

By using the above matrix $\check S_q$, we provide the following important lemma for Markov parameters $\check H(i)$ of the cycled system\cite{okajima}.
\begin{lemma}\label{lemma3}
Consider the following $Ml\times Mm$ matrix.
\begin{eqnarray}
 \check S_l^i \check H(i+j)\check S_m^{j}
\end{eqnarray}
Then, $\check S_l^i \check H(i+j)\check S_m^j$ can be regarded as a block diagonal matrix with $l\times m$ block elements for any non-negative integers $i, j$. In addition, the following matrix: 
\begin{eqnarray}
 \check S_l^i \check H(i+j)\check S_m^{j-1}
\end{eqnarray}
can be regarded as a cyclic matrix. %\hfill $\Box$
\end{lemma}
Lemma \ref{lemma3} provides valuable properties that hold for the cyclic reformulations. Additionally, the following two lemmas are derived using Lemma \ref{lemma3}.

\begin{lemma}\label{lemma1}
Consider the following $Ml\times Mm$ matrix.
\begin{eqnarray}
\check S_l^i \check H(i)\label{mar1}
\end{eqnarray}
Then, $\check S_l^i \check H(i)$ is given as a block diagonal matrix with $l\times m$ block elements for any non-negative integer $i$. In addition, the following matrix: 
\begin{eqnarray}
\check S_l^{i-1} \check H(i) \label{mar2}
\end{eqnarray}
can be regarded as a cyclic matrix. 
\end{lemma}

\begin{lemma}\label{lemma2}
Consider the following $Ml\times Mm$ matrix.
\begin{eqnarray}
 \check H(i)\check S_m^i
\end{eqnarray}
Then, $ \check H(i)\check S_m^i$ can be regarded as a block diagonal matrix with $l\times m$ block elements for any non-negative integer $i$. In addition, the following matrix: 
\begin{eqnarray}
 \check H(i)\check S_m^{i-1}
\end{eqnarray}
can be regarded as a cyclic matrix. %\hfill $\Box$
\end{lemma}

%In addition to Lemmas \ref{lemma1} and \ref{lemma2}, the following lemma can be obtained. 

The characteristics shown in Lemmas \ref{lemma1} and \ref{lemma2} are essential ideas for identifying the $M$-periodic multirate system and are used later. 

%Note that Lemmas \ref{lemma1} and \ref{lemma2} are automatically satisfied if we select appropriate $i$ or $j$ in Lemma \ref{lemma3}. Therefore, it is sufficient to handle Lemma \ref{lemma3} as a property of the system with cyclic reformulation. 
In addition to the above lemmas, we provide useful properties for systems with cyclic reformulation. 
First, matrices $\check F_j$ of size $Mn\times Ml$ are determined as follows. 
%\begin{eqnarray}
%\check F_j = \bf{diag}\{F_j,\cdots,F_j\}. 
%\end{eqnarray}
\begin{equation}
     \check{F}_j = \left[\begin{array}{cccc}
          F_j & O_{n,l} & \cdots & O_{n,l}\\
          O_{n,l} & \ddots & \ddots & \vdots\\
          \vdots & \ddots  & F_j & O_{n,l}\\
          O_{n,l} & \cdots & O_{n,l} & F_j
     \end{array} \right]\label{fj}
\end{equation}
Elements $F_j$ are $n \times l$ matrices as presented in (\ref{largef}). 
We assume $F_j (j=0,\cdots,n-1)$ is determined to satisfy the following condition.
\begin{eqnarray}\label{rankf}
{\bf{rank}} \left[\begin{array}{ccccc}
F_0 & \cdots& F_i & \cdots & F_{n-1}
\end{array}\right]  = n
\end{eqnarray}
By using matrices $\check S_l$ in (\ref{sq}) and $\check F_j$ in (\ref{fj}), a matrix $\check X$ can be derived by the following calculation with cycled system parameters: 
\begin{eqnarray}
\check X = \sum_{i=0}^{n-1}\sum_{j = 0}^{M-1} \check F_j \check S_l^{j} \check C \check A^{Mi+j}. \label{checkx2}
\end{eqnarray}
The matrix rank of $\check X$ is $Mn$ if each $X_i$ in (\ref{fk}) is given as a regular matrix. The matrix rank of $\check X$ depends on the observability of the system, which is presented in (\ref{psij}), and the selection of $F$. Note that $\check X\check B$ has a cyclic matrix structure. 

Then, a matrix $\check G_j$ with the size ($Mm \times Mn$) is determined as follows. 
\begin{equation}
     \check{G_j} = \left[\begin{array}{cccc}
          G_j & O_{m,n} & \cdots & O_{m,n}\\
          O_{m,n} & \ddots & \ddots & \vdots\\
          \vdots & \ddots  & G_j & O_{m,n}\\
          O_{m,n} & \cdots & O_{m,n} & G_j
     \end{array} \right]. \label{gj}
\end{equation}
Elements $G_j$ are $m \times n$ matrices. 
We assume $G_j (j=0,\cdots,n-1)$ is determined to satisfy the following condition.
\begin{eqnarray}\label{rankg}
{\bf{rank}} \left[\begin{array}{c}
G_0\\ \vdots\\ G_i \\ \vdots \\ G_{n-1}
\end{array}\right]  = n
\end{eqnarray}
By using the above matrices $\check G_j$, a matrix $\check Y$ is obtained with $\check A, \check B$ and $\check S_m$. 
\begin{eqnarray}
\check Y = \sum_{i = 0}^{n-1}\sum_{j = 0}^{M-1}  \check A^{Mi+j}\check B \check S_m^{j+1} \check G_j  \label{checky2}
\end{eqnarray}
The matrix size of $\check Y$ is $Mn \times Mn$ and is a full-rank matrix. Note that $\check C\check Y$ is a block diagonal matrix structure. 
%%%%%%%%%%%%%%%%%%%%%%%%%%%%%%%%%%%%%%%%%%%%%%%
\section{System Identification Algorithm}\label{sec4}
%%%%%%%%%%%%%%%%%%%%%%%%%%%%%%%%%%%%%%%%%%%%%%%
%＝＝＝＝＝＝＝＝＝＝＝＝＝＝＝＝＝＝＝＝＝＝＝＝＝＝＝＝＝＝＝＝＝＝＝＝＝＝＝＝＝＝＝＝
\subsection{System identification using cycled signals}\label{sec41} 
%＝＝＝＝＝＝＝＝＝＝＝＝＝＝＝＝＝＝＝＝＝＝＝＝＝＝＝＝＝＝＝＝＝＝＝＝＝＝＝＝＝＝＝＝
We use the subspace identification method since we employ state-space models for parameter identification. The subspace identification method \cite{id1,id2} is a significant system identification method based on the state-space realization of linear time-invariant systems. The advantage of using the subspace identification method is that it can be easily applied to MIMO systems, and they use numerically stable algorithms such as singular value decomposition and QR decomposition. Therefore, it is possible to obtain an accurate state-space model. 

We apply the system identification method as follows: First, we apply input $u(k)$ for (\ref{eq:P_multi}) and obtain an output signal $y(k)$. Then, the cyclic reformulation is applied to the input and output data $(u,y)$, and we obtain cycled signals $\check u(k) \in R^{Mm}$ and $\check y(k) \in R^{Ml}$. Moreover, the subspace identification method is applied for the cycled signals $\check u(k), \check y(k)$ and obtain state-space model parameters ($\A_*,\B_*,\C_*,\D_*$). The matrix sizes are $\A_* \in R^{Mn\times Mn}$, $\B_* \in R^{Mn\times Mm}$, $\C_* \in R^{Ml\times Mn}$, $\D_* \in R^{Ml\times Mm}$. The state of the identified state-space model is denoted as $x_* \in R^{Mn}$. 

%System identification is performed based on the subspace identification method using the cycled signals $\check u(k)$ and $\check y(k)$. Then, a system obtained by the system identification method is denoted as $\A_*,\B_*,\C_*,\D_*$. 
The Markov parameters for the obtained system by the subspace identification are given below. 
\begin{eqnarray}
\check \Hm (i) = \begin{cases}\D_*,&i = 0 \\ \C_*  \A_*^{i-1} \B_*, & i=1,2,\cdots \end{cases} 
\end{eqnarray}
The following $Ml\times Mm$ matrix is considered in the same manner as in Lemma \ref{lemma3}. 
\begin{eqnarray}
\check S_l^{i} \check \Hm (i+j) \check S_m^j. \label{mar12}
\end{eqnarray}

Then, we set up the following assumption related to Lemma \ref{lemma3} for the state-space model parameters ($\A_*,\B_*,\C_*,\D_*$). 

\begin{ass}\label{ass1}
For any $i, j (= 0,1,\cdots)$, the matrix $ \check S_l^{i} \check \Hm (i+j) \check S_m^j$ is a block diagonal matrix with $l\times m$ block elements, where all off-diagonal blocks are zero matrices. $\hfill  \Box$
\end{ass}

We verify that Assumption \ref{ass1} is reasonable through a numerical simulation. 
In the simulation, $n = 3, m = 1, l = 2$ is selected, and the following parameter $(A, B, C, D)$ is considered as a plant. 
\begin{eqnarray}
&A = \begin{bmatrix}0&0&0.8\\1&0&0.5\\0&1&-0.4\end{bmatrix}, B = \begin{bmatrix}1\\0\\0\end{bmatrix}, \label{rei01} \\ &C = \begin{bmatrix}1&0.5&0.3\\0.1&0.3&0.7\end{bmatrix}, D = \begin{bmatrix}0\\0\end{bmatrix}\label{rei02}
\end{eqnarray}
In addition, $M_0=1$ and $M_1 = 3$ are selected, and $V_0, V_1, V_2$ is given as follows.
\begin{eqnarray}
&V_0 = \begin{bmatrix}1&0\\0&1\end{bmatrix}, V_1 = \begin{bmatrix}1&0\\0&0\end{bmatrix}, V_2 = \begin{bmatrix}1&0\\0&0\end{bmatrix}
\end{eqnarray}
The cyclic reformulation of the multirate system can be written as follows. 
\begin{eqnarray}
     &\check{A} = \left[\begin{array}{ccc}
          O_{3,3} & O_{3,3} & A\\
          A &  O_{3,3} & O_{3,3}\\
          O_{3,3} & A &  O_{3,3}
     \end{array} \right], \\&\check{B} = \left[\begin{array}{ccc}
          O_{3,1} & O_{3,1} & B\\
          B &  O_{3,1} & O_{3,1}\\
          O_{3,1} & B &  O_{3,1} 
     \end{array} \right], 
\end{eqnarray}
\begin{equation}
     \check{C} = \left[\begin{array}{ccc}
          V_0 C & O_{2,3} &O_{2,3} \\
          O_{2,3}  & V_1 C & O_{2,3}\\
          O_{2,3} & O_{2,3}  & V_{2}C
     \end{array} \right], \check{D} = O_{6,3}. \nonumber
\end{equation}
Since $M = 3$ holds, the matrix $\check S_1$ is given by
\begin{equation}
     \check S_1 = \left[\begin{array}{ccccc}
          0 & 1 & 0\\
          0 & 0 & 1\\
          1 & 0 & 0 
     \end{array} \right]. 
\end{equation}
By calculating (\ref{mar1}), the parameters $\check H(i) \check S_1^i $ for the plant $P_{ex}$ can be calculated as follows. 
\begin{eqnarray}
&\check H (0) = O_{6,3}, \\ 
     &\check H (1) \check S_1 = \left[\begin{array}{ccc}
          1 & 0 & 0 \\
          0.1 & 0 & 0\\
          0 & 1 & 0\\
          0 & 0 & 0 \\
          0 & 0 & 1\\
          0 & 0 & 0
     \end{array} \right], \\ 
     &\check H (2)\check S_1^2  = \left[\begin{array}{ccc}
          0.5 & 0 & 0 \\
          0.3 & 0 & 0\\
          0 & 0.5 & 0\\
          0 & 0 & 0 \\
          0 & 0 & 0.5\\
          0 & 0 & 0
     \end{array} \right], \\&\check H (3) \check S_1^3 = \left[\begin{array}{ccc}
           0.3 & 0 & 0 \\
          0.7 & 0 & 0\\
          0 & 0.3 & 0\\
          0 & 0 & 0 \\
          0 & 0 & 0.3\\
          0 & 0 & 0
     \end{array} \right], \\
     &\check S_1^4 \check H (4) =  \left[\begin{array}{ccc}
           0.93 & 0 & 0 \\
          0.05 & 0 & 0\\
          0 & 0.93 & 0\\
          0 & 0 & 0 \\
          0 & 0 & 0.93\\
          0 & 0 & 0
     \end{array} \right] \cdots 
\end{eqnarray}
We can see that $\check H(i)\check S_1^i $ is given as diagonal matrices as indicated in Lemma \ref{lemma2}. Due to space constraints, the proof is omitted; however, it is evident that Lemma \ref{lemma3} holds in a similar manner.

By applying an input $u(k)$ to the above system, we obtain outputs $y(k) = [y_0(k), y_1(k)]^T$ of the plant. The rate of $y_0(k)$ is 1 and that of $y_1(k)$ is $3$. Note that the input $u(k)$ is randomly selected for each time step and is not an $M$-periodic signal. Cycled signals $\check u(k)$ and $\check y(k)$ with $M=3$. Then, a subspace identification method is applied for $\check u(k)$ and $\check y(k)$. We use the N4SID method \cite{id1}, which is a kind of subspace identification method, in this simulation. The N4SID method is equipped with the "System Identification Toolbox" on MATLAB and can be easily implemented. 

The parameters $\A_*, \B_*, \C_*, \D_*$ are obtained by using the N4SID method. $\A_* \in R^{9\times 9}$, $\B_* \in R^{9\times 3}$, $\C_* \in R^{6\times 9}$, $\D_* \in R^{6\times 3}$. Then, we calculate $\check S_1^i \check \Hm (i), i = 0, \cdots, 4$ for obtained $\A_*, \B_*, \C_*, \D_*$ as follows. 
\begin{eqnarray}
&\check \Hm (0) = O_{6,3}, \\& \check \Hm  (1) \check S_1 = \left[\begin{array}{ccc}
          1.000 & 0.000 & 0.000 \\
          0.100 & 0.000 & 0.000\\
          0.000 & 1.000 & 0.000\\
          0.000 & 0.000 & 0.000 \\
          0.000 & 0.000 & 1.000\\
          0.000 & 0.000 & 0.000
     \end{array} \right], 
     \end{eqnarray}
     \begin{eqnarray}
     &\check \Hm (2)\check S_1^2  = \left[\begin{array}{ccc}
          0.500 & 0.000 & 0.000 \\
          0.300 & 0.000 & 0.000\\
          0.000 & 0.500 & 0.000\\
          0.000 & 0.000 & 0.000 \\
          0.000 & 0.000 & 0.500\\
          0.000 & 0.000 & 0.000
     \end{array} \right], \\
     &\check \Hm (3) \check S_1^3 = \left[\begin{array}{ccc}
           0.300 & 0.000 & 0.000 \\
          0.700 & 0.000 & 0.000\\
          0.000 & 0.300 & 0.000\\
          0.000 & 0.000 & 0.000 \\
          0.000 & 0.000 & 0.300\\
          0.000 & 0.000 & 0.000
     \end{array} \right], \\
     &\check \Hm (4)\check S_1^4  =  \left[\begin{array}{ccc}
           0.930 & 0.000 & 0.000 \\
          0.050 & 0.000 & 0.000\\
          0.000 & 0.930 & 0.000\\
          0.000 & 0.000 & 0.000 \\
          0.000 & 0.000 & 0.930\\
          0.000 & 0.000 & 0.000
     \end{array} \right], \cdots 
\end{eqnarray}
By checking each matrix $\check S_1^i \check \Hm (i)$, we can confirm that there are diagonal matrices for all $i$ in this simulation result. Although not shown here, it is confirmed that the diagonal matrix can be obtained in the same way when $i$ is $5$ or more. Therefore, we can confirm that Assumption \ref{ass1} holds. Furthermore, we can also confirm that $\check S_1^i\check \Hm (i)$ coincides with $\check S_1^i\check H(i)$ for $i=0,\cdots,4$. Due to space constraints, the details are omitted; however, all matrices except for $\D_*$ are obtained as dense matrices. 

\subsection{Transformation of obtained state-space model}
%\begin{eqnarray}
%\check \A = T^{-1}\A_* T, \check \B = T^{-1}\B_*, \check \C = \C_* T, \check \D = \D_* \label{henkan}
%\end{eqnarray}
%Note that the matrix $T$ must be invertible. New state $\check x_{tf} \in {\bf{R}}^{Mn\times 1}$ is given by $\check x_{tf} = T^{-1} \check x_*$. 
%The following state-space model is obtained using the transformation matrix $T$. 
%\begin{equation}
%     \label{cycsolution}
%     \begin{array}{rcl}
%          \check{x}_{tf}(k+1) & = & \check{\A}\check{x}_{tf}(k) + \check{\B}\check{u}(k)\\
%          \check{y}(k) & = & \check{\C}\check{x}_{tf}(k) + \check{\D}\check{u}(k) 
%     \end{array}
%\end{equation}
%By selecting an appropriate $T$ for the obtained $\A_*, \B_*, \C_*, \D_*$ in system identification, the objective of this study, which is presented in %Problem \ref{prob1}, will be achieved if $\check \A, \check \B, \check \C, \check \D$ are given as a cyclic reformulation form. 

%Using the matrices obtained as described above, $T^{-1}$ is defined as follows, where $T^{-1}$ is the inverse matrix of the transformation matrix $T$. 
%\begin{eqnarray}
%T^{-1} = \sum_{j = 1}^n \check F_j \check S_l^{j-1} \C_* \A_{*}^{j-1}\label{coordinate}
%\end{eqnarray}
%The matrix form of $\check F_j$ is given in (\ref{fj}). The following theorem holds for the state coordinate transformation matrix $T$ given as the inverse of (\ref{coordinate}) for the case that $T$ is a regular matrix. Also, we can see that (\ref{coordinate}) and (\ref{checkx2}) are closely related. 

The matrix parameters are obtained as $\A_*, \B_*, \C_*, \D_*$ by using the subspace identification method with the cycled signals $\check u(k)$ and $\check y(k)$. Unfortunately, $\A_*, \B_*, \C_*, \D_*$ are expected to be dense matrices and are not obtained as a cyclic reformulation structure. A state coordinate transformation for the system parameters ($\A_*$, $\B_*$, $\C_*$, $\D_*$) using the specified transformation matrix $T \in {R}^{Mn\times Mn}$ is considered for obtaining cyclic reformulation in this section. The transformation matrix is derived based on Assumption \ref{ass1}.

The state-space vector $\check x_* \in {R}^{Mn\times 1}$ of a state-space model $\A_*, \B_*, \C_*, \D_*$ is determined. 
The transformation matrix $T$ is set to give the matrices as follows. 
\begin{eqnarray}
%\begin{split}
\check A_m&=T^{-1}\A_* T&, \check B_m=T^{-1}\B_*\label{eq:T1}\\
\check C_m&=\C_*T&, \check D_m=\D_*\label{eq:T2}
%\end{split}
\end{eqnarray}
The following state-space model is obtained using the transformation matrix $T$. 
\begin{equation}
     \label{cycsolution}
     \begin{array}{rcl}
          \check{x}_{tf}(k+1) & = & \check A_m\check{x}_{tf}(k) + \check B_m\check{u}(k)\\
          \check{y}(k) & = & \check C_m \check{x}_{tf}(k) + \check{D}_m\check{u}(k) 
     \end{array}
\end{equation}
The new state $\check x_{tf} \in {R}^{Mn\times 1}$ is given by $\check x_{tf} = T^{-1} \check x_*$. 

By selecting an appropriate $T$ for the obtained parameters $\A_*, \B_*, \C_*, \D_*$, the objective presented in Problem \ref{prob1} will be achieved if $\check A_m, \check B_m, \check C_m, \check D_m$ are given as a cyclic reformulation form as follows. 
\begin{eqnarray}
\label{Amtilde}
\check A_m=\left[\begin{array}{ccccc}
O_{n,n} & \cdots & \cdots & O_{n,n} & A_{m(M-1)} \\
A_{m0} & O_{n,n} & \cdots & O_{n,n} & O_{n,n} \\
O_{n,n} & A_{m1} & \ddots & \vdots & \vdots \\
O_{n,n} & \ddots & \ddots & O_{n,n} & \vdots \\
O_{n,n} & \cdots & O_{n,n} & A_{m(M-2)} & O_{n,n}
\end{array}\right]
\end{eqnarray}
\begin{eqnarray}\label{Bmtilde}
\check{B}_{m}=\left[\begin{array}{ccccc}
O_{n,m} & \cdots & \cdots & O_{n,m} & B_{m(M-1)} \\
B_{m0} & O_{n,m} & \cdots & O_{n,m} & O_{n,m} \\
O_{n,m} & B_{m1} & \ddots & \vdots & \vdots \\
O_{n,m} & \ddots & \ddots & O_{n,m} & \vdots \\
O_{n,m} & \cdots & O_{n,m} & B_{m(M-2)} & O_{n,m}
\end{array}\right]
\end{eqnarray}
\begin{eqnarray}
\label{Cmtilde}
\check C_{m}=\mathrm{diag}\left[\begin{array}{c}
C_{m0}, C_{m1}, \cdots, C_{m{M-1}}
\end{array}\right]
\end{eqnarray}
\begin{eqnarray}
\label{Dmtilde}
\check D_{m}=\mathrm{diag}\left[\begin{array}{c}
D_{m0}, D_{m1}, \cdots, D_{m(M-1)}
\end{array}\right]
\end{eqnarray}
Using the matrices obtained as described above, a coordinate transform matrix $T$ in this paper is defined as follows. 
\begin{eqnarray}
T = \sum_{i = 0}^{n-1}\sum_{j = 0}^{M-1}  \A_*^{Mi+j}\B_* \check S_m^{j+1} \check G_{(Mi+j) \bmod n}  \label{checky3}
\end{eqnarray}
As shown previously, it obviously satisfies $ \check G_i = \check S_m \check G_{(Mi+j) \bmod n}  \check S_n^{-1}$. $T$ in (\ref{checky3}) is equivalent to (\ref{checky4}). 
\begin{eqnarray}
T = \sum_{i = 0}^{n-1}\sum_{j = 0}^{M-1}  \A_*^{Mi+j}\B_*  \check G_{(Mi+j) \bmod n}  \check S_n^{j+1} \label{checky4}
\end{eqnarray}

%このとき，$\check S_q^{i}\check C \check S_n^{-i}$によって$i=0,\cdots, M-1$にて要素をシフトできるため，可観測性の条件を満たしていなくても$T^{-1}$をフルランク構造にでき，適切にサイクリング構造を求めることができる． 
%ここで，$j = 1，\cdots，n$であり，$\tilde{G}(j)$の行列内要素$g(j)\in R^{m\times n}$を(\ref{eq:g(j)})に示す．
% 
%\begin{eqnarray}\label{eq:g(j)}
%g(j)=\left[\begin{array}{c}
%g_1，\cdots， g_i， \cdots， g_n
%\end{array}\right] 
%\end{eqnarray}
%このとき，$i$番目の行列内要素$g_i\in R^{m\times 1}$は（\ref{eq:gi}）で表される．
%
%\begin{eqnarray}\label{eq:gi}
%g_i = \left\{ \begin{array}{l} 
%\displaystyle 1\;\;\;\;\;(i=j) \\
%0\;\;\;\;\;(i\neq j)
%\end{array} \right.
%\end{eqnarray}
The matrix form of $\check F_j$ is given in (\ref{fj}). We should appropriately select $\check F_j (j= 1,\cdots, n)$ based on the rank condition. 
The following theorem holds for the state coordinate transformation matrix $T$ given in (\ref{checky3}) for the case that $T$ is a regular matrix. Also, we can see that (\ref{checky2}) and (\ref{checky3}) are closely related. 

The following theorem is presented to obtain the cyclic reformulation of the derived model. 
%\begin{theorem}\label{theo1}
%Assuming that the parameters $\A_*, \B_*, \C_*, \D_*$ are obtained via the subspace identification based on the cycled signal. Then, Assumption \ref{ass1} is satisfied for $\A_*, \B_*, \C_*, \D_*$. In addition, the pairs $(\A_*, \B_*)$ and $(\C_*, \A_*)$ are controllable and observable, respectively. Then, the system $\check A_m, \check B_m, \check C_m, \check D_m$, which is obtained by the state coordinate transformation (\ref{eq:T1}), (\ref{eq:T2}) of  $\A_*, \B_*, \C_*, \D_*$ using the transformation matrix $T$ of (\ref{checky3}), has a structure of the cyclic reformulation. 
%\end{theorem}

\begin{theorem}\label{theo1}
Let $\A_*, \B_*, \C_*, \D_*$ be parameters based on the cycled signal. Assume that Assumption \ref{ass1} is satisfied for these parameters and that the pairs $(\A_*, \B_*)$ and $(\C_*, \A_*)$ are controllable and observable, respectively.
Then the system $\check A_m, \check B_m, \check C_m, \check D_m$, obtained by the state coordinate transformation (\ref{eq:T1}), (\ref{eq:T2}) of $\A_*, \B_*, \C_*, \D_*$ using the transformation matrix $T$ of (\ref{checky3}), has a structure of the cyclic reformulation.
\end{theorem}

\begin{proof}
See Appendix. 
\end{proof} 

When Assumption \ref{ass1} is satisfied by the subspace identification with cycled signals, the obtained model with the coordinate transform using (\ref{checky3}) is given as cyclic reformulation structure by Theorem \ref{theo1}. 
While there is freedom for the coordinate transformations of each parameter matrix that is given as a cyclic reformulation, in addition to the obtained $T^{-1}$, this degree of freedom for the coordinate transformation can be achieved by using $T\Phi$ instead of $T$, where $\Phi$ is a block diagonal structure matrix given as:  
\begin{eqnarray}
\Phi = \left[\begin{array}{cccc}
          \Phi_1 & O_{n,n} & \cdots & O_{n,n}\\
          O_{n,n} & \ddots & \ddots & \vdots\\
          \vdots & \ddots  & \Phi_{M-1} & O_{n,n}\\
          O_{n,n} & \cdots & O_{n,n} & \Phi_M
     \end{array} \right],
\end{eqnarray}
where $\Phi_i (i=1,\cdots,M)$ should be regular matrices. 

Note that the coordinate transformation matrix of this paper is different from the case of the LPTV system shown in \cite{okajima}. In the LPTV system of \cite{okajima}, the system is assumed to be observable and controllable at each time instant. However, multirate systems typically do not satisfy this assumption since observability and controllability properties vary across sampling instants due to the inherent structure of multirate sensing. Consequently, if we attempt to directly apply the method from \cite{okajima} to multirate systems, the resulting coordinate transformation matrix $T$ would be singular, making it impossible to properly determine the cycling structure. Our proposed approach overcomes this fundamental limitation by developing a new coordinate transformation method tailored explicitly for multirate systems that ensures the regularity of the transformation matrix. 

\subsection{System identification algorithm}\label{sec43}

In the previous sections, we explained the identification using cycled signals and the coordinate transformation to realize the cycling structure.
Consequently, the cyclic identification algorithm for multirate sensing systems is summarized as following Algorithm \ref{algo11}. 

\begin{algorithm}
\caption{System Identification for multirate sensing systems}
\label{algo11}
\begin{algorithmic}
%\begin{itemize}
\STATE [1.] Determine $M$ from each output rate $M_i$ of the multirate system and prepare cycled input and cycled output signals by the obtained input-output data from the system (\ref{eq:P_multi}), (\ref{eq:P_multi2}). 
\STATE [2.] Compute $\A_*$, $\B_*$, $\C_*$,$\D_*$ using the existing subspace identification method with the cycled signals. 
\STATE [3.] Cyclic reformulation is derived using the obtained $\A_*$, $\B_*$, $\C_*$, $\D_*$ with the specific state coordinate transformation matrix $T$ from (\ref{checky3}). 
\STATE [4.] Parameters $A_{mi}, B_{mi}, C_{mi}, D_{mi}$ are extracted from the components of the cyclic reformulation $\check A_m, \check B_m, \check C_m, \check D_m$. 
\end{algorithmic}
\end{algorithm}

The system identification process for multirate sensing systems using Algorithm \ref{algo11} proceeds as follows. In Step 1, the input-output signals of the multirate sensing system are transformed into a cyclic structure, resulting in cycled input-output signals. In Step 2, these cycled signals are applied to a subspace identification method to determine $\A_*$, $\B_*$, $\C_*$, and $\D_*$. The dimensions of the resulting matrices are as follows: $\A_*\in {R}^{Mn\times Mn}$, $\B_*\in {R}^{Mn\times Mm}$, $\C_*\in {R}^{Ml\times Mn}$, and $\D_*\in {R}^{Ml\times Mm}$. Based on Assumption \ref{ass1}, $T$ is derived by (\ref{checky3}) and apply coordinate transformation to the system ($\A_*$, $\B_*$, $\C_*$, and $\D_*$) in Step 3. Then, we obtain $\check A_m, \check B_m, \check C_m, \check D_m$. Since $\check A_m, \check B_m, \check C_m, \check D_m$ are shown to be cyclic reformulations in Theorem \ref{theo1}, the matrix elements of $\check A_m, \check B_m, \check C_m, \check D_m$ can be extracted to obtain $A_{mi}, B_{mi}, C_{mi}, D_{mi}$.

Steps 1 and 3 require minimal computation time. In Step 2, the cycled signals are used for parameter estimation through conventional subspace identification methods, enabling the solution to be obtained with computational time comparable to that required for identifying a linear time-invariant system. 

\section{Numerical Example}\label{sec5}

\subsection{Setting of the multirate system}

In this section, numerical simulations of the proposed system identification algorithm are verified. Plant parameters, shown in (\ref{rei01}), (\ref{rei02}) are used in this section. In addition, the rates of the sensor outputs are given as $M_0=2$ and $M_1 = 3$. $M = 6$ is derived from these rates. In addition, $V_i (i=0,\cdots,5)$ are given based on the rates as follows. 
%Sub-Section 2-1.
\begin{eqnarray}
&V_0 = \begin{bmatrix}1&0\\0&1\end{bmatrix}, V_1 = \begin{bmatrix}0&0\\0&0\end{bmatrix}, V_2 = \begin{bmatrix}1&0\\0&0\end{bmatrix}, \\
&V_3 = \begin{bmatrix}0&0\\0&1\end{bmatrix}, V_4 = \begin{bmatrix}1&0\\0&0\end{bmatrix}, V_5 = \begin{bmatrix}0&0\\0&0\end{bmatrix} 
\end{eqnarray}
The pair $(A,B)$ is controllable and $(V_0 C, A)$ is observable. Then, ${\bf{rank}} {\Psi}_c = 18$ and ${\bf{rank}} {\Psi}_o = 18$ hold for the considered system. We can see that Assumption \ref{ass2} is satisfied for the plant. 

Note that if the model parameters $(A, B, C, D)$ in (\ref{rei01}), (\ref{rei02}) are regarded as an LTI system, the transfer function for each input/output can be obtained as follow. 
\begin{eqnarray}
  &TF_1 = \frac{z^2+0.9z}{z^3+0.4z^2-0.5z-0.8}\\ &TF_2 = \frac{0.1z^2+0.34z+0.77}{z^3+0.4z^2-0.5z-0.8}
\end{eqnarray}
In the later part of this section, $TF_1(z)$ and $TF_2(z)$ are used to verify the estimation accuracy of the proposed system identification algorithm. 

\subsection{Simulation result}

%さらに，サイクリングが条件を満たすことも記載する．(19)の条件を満たす旨を記載する．

By applying an input $u(k)$ for the multirate system, we obtain an output $y(k)$ of the plant. Note that the input $u(k)$ in the simulation is randomly selected for each time step and is not a $6$-periodic signal. 

Cycled signals $\check u(k) \in R^{6}$ and $\check y(k) \in R^{12}$ is obtained from $u(k)$ and $y(k)$. Then, a subspace identification method is applied for the obtained $\check u(k)$ and $\check y(k)$. We use the N4SID method \cite{id1}, which is a kind of subspace identification method, in this simulation. The parameters are derived as $\A_* \in R^{18\times 18}, \B_* \in R^{18\times 6}, \C_* \in R^{12\times 18}, \D_* \in R^{12\times 6}$ by using the N4SID method. Markov parameters $\check \Hm (i)$ are obtained from the parameters $( \A_*, \B_*, \C_*, \D_* )$. Due to space constraints, we omitted the results, but we confirmed that $\|H(i) -\Hm(i)\|$ is very small for all $i$. Lemma \ref{lemma3} is satisfied for $H(i)$. Therefore, Assumption \ref{ass1} is satisfied for the obtained system by the subspace identification method. 

We give $G_i$ as follows. 
\begin{eqnarray}
&G_0 = \begin{bmatrix}1&0&0\end{bmatrix}, G_1 = \begin{bmatrix}0&1&0\end{bmatrix},\\& G_2 = \begin{bmatrix}0&0&1\end{bmatrix}
  \end{eqnarray}
$G:=[G_0^T,G_1^T,G_2^T]^T$ is an identity matrix and is full-rank. 
$\check G_i$ is given to satisfy (\ref{rankg}) and their matrix size are $6 \times 18$. 

By applying step 3. in Algorithm \ref{algo11}, $\A_*, \B_*, \C_*$ and $\D_*$ are transformed by the following coordinate transform matrix $T$.
\begin{eqnarray}
T = \sum_{i = 0}^{2}\sum_{j = 0}^{5}  \A_*^{6i+j}\B_* \check S_1^{j+1} \check G_i  \label{henkanTinv}
\end{eqnarray}
The matrix rank of $T$ is 18. Since $T$ is given as a regular matrix, the matrix $T^{-1}$ is obtained as the inverse matrix of (\ref{henkanTinv}). 

Furthermore, the following matrices are obtained by applying a state coordinate transformation matrix (\ref{henkanTinv}) to the obtained $(\A_*,\B_*,\C_*,\D_*)$ by the subspace identification. We obtain $(\check A_m, \check B_m, \check C_m, \check D_m)$ and satisfy the cyclic reformulation structure. Thus, having obtained the cyclic reformulation, we can obtain $A_{mi}, B_{mi}, C_{mi}, D_{mi}, i=0,\cdots, 5$ as their elements by Step 4 in Algorithm \ref{algo11}. 
\begin{eqnarray}
&A_{m0} = \begin{bmatrix}1.0129 & -2.0947 & 2.3008\\0.8062 & -0.5788 & 1.3884\\-0.5685 & 1.8355 & -0.8341\end{bmatrix},\label{a0} \\&A_{m1} = \begin{bmatrix}1.0129 & -2.0947 & 2.3008\\0.8062 & -0.5788 & 1.3884\\-0.5685 & 1.8355 & -0.8341\end{bmatrix}, \\
&A_{m2} = \begin{bmatrix}1.0129 & -2.0947 & 2.3008\\0.8062 & -0.5788 & 1.3884\\-0.5685 & 1.8355 & -0.8341\end{bmatrix}, \\&A_{m3} = \begin{bmatrix}1.0129 & -2.0947 & 2.3008\\0.8062 & -0.5788 & 1.3884\\-0.5685 & 1.8355 & -0.8341\end{bmatrix}, \\
&A_{m4} = \begin{bmatrix}1.0129 & -2.0947 & 2.3008\\0.8062 & -0.5788 & 1.3884\\-0.5685 & 1.8355 & -0.8341\end{bmatrix}, \\&A_{m5} = \begin{bmatrix}1.0129 & -2.0947 & 2.3008\\0.8062 & -0.5788 & 1.3884\\-0.5685 & 1.8355 & -0.8341\end{bmatrix},
\end{eqnarray}
\begin{eqnarray}
&B_{m0} = \begin{bmatrix}   -0.5783 \\  -0.7672\\0.8871\end{bmatrix}, B_{m1} = \begin{bmatrix}   -0.5783 \\  -0.7672\\0.8871\end{bmatrix}, \label{b0}
\\ &B_{m2} = \begin{bmatrix}   -0.5783 \\  -0.7672\\0.8871\end{bmatrix}, B_{m3} = \begin{bmatrix}   -0.5783 \\  -0.7672\\0.8871\end{bmatrix}, 
\\ &B_{m4} = \begin{bmatrix}   -0.5783 \\  -0.7672\\0.8871\end{bmatrix}, B_{m5} = \begin{bmatrix}   -0.5783 \\  -0.7672\\0.8871\end{bmatrix}, 
\end{eqnarray}
\begin{eqnarray}
&C_{m0} = \begin{bmatrix} 1.8058 & 2.3016 & 4.2947 \\ 0.6785 & 2.1105 & 2.3801\end{bmatrix}, \\ &C_{m1} = \begin{bmatrix} 0.0000 & 0.0000 & 0.0000 \\ -0.0000 & -0.0000 & -0.0000\end{bmatrix}, \label{c0} \\ 
&C_{m2} = \begin{bmatrix} 1.8058 & 2.3016 & 4.2947 \\  -0.0000 & -0.0000 & -0.0000\end{bmatrix}, \\ &C_{m3} = \begin{bmatrix} 0.0000 & 0.0000 & 0.0000 \\ 0.6785 & 2.1105 & 2.3801\end{bmatrix},  \\ 
&C_{m4} = \begin{bmatrix} 1.8058 & 2.3016 & 4.2947 \\ -0.0000 & -0.0000 & -0.0000\end{bmatrix}, \\ &C_{m5} = \begin{bmatrix} 0.0000 & 0.0000 & 0.0000 \\ -0.0000 & -0.0000 & -0.0000\end{bmatrix}, 
\end{eqnarray}
\begin{eqnarray}
&D_{m0} = \begin{bmatrix} 0.0000 \\ 0.0000\end{bmatrix}, D_{m1} = \begin{bmatrix} 0.0000 \\ 0.0000\end{bmatrix}, \label{d0}\\
&D_{m2} = \begin{bmatrix} 0.0000 \\ 0.0000\end{bmatrix}, D_{m3} = \begin{bmatrix} 0.0000 \\ 0.0000\end{bmatrix},  \\
&D_{m4} = \begin{bmatrix} 0.0000 \\ 0.0000\end{bmatrix}, D_{m5} = \begin{bmatrix} 0.0000 \\ 0.0000\end{bmatrix}.
\end{eqnarray}

%(以下、未修正) A, B, C, Dがそれぞれ全ての時刻で一致していることを記載する．

The model derived using Algorithm \ref{algo11} requires minimal computation time, comparable to that of identifying linear time-invariant systems, making it an efficient and practical approach. We can verify that $A_{m0} = A_{m1} = \cdots = A_{m5}$ is satisfied numerically. 
%Consequently, Problem \ref{prob1} is successfully solved in the noise-free case. One of the key advantages of our proposed system identification algorithm is its capability to obtain a multirate model without depending on any specific periodic input signals. 
%伝達関数が一致することの確認
Due to the degrees of freedom in coordinate transformation, the obtained matrix differs from the predefined plant parameters. Therefore, we consider a comparison based on transfer functions to demonstrate that the system has been correctly identified. When $(A_{m0}, B_{m0}, C_{m0}, D_{m0})$ from (\ref{a0}), (\ref{b0}), (\ref{c0}), (\ref{d0}) are considered as LTI system parameters, the transfer function can be obtained as follows. 
\begin{eqnarray}
  &TF_{m1} = \frac{z^2+0.9z+1.29\times 10^{-15}}{z^3+0.4z^2-0.5z-0.8} \\ &TF_{m2} = \frac{0.1z^2+0.34z+0.77}{z^3+0.4z^2-0.5z-0.8}
\end{eqnarray}
We can find that $TF_{m1}$ and $TF_1$ are infinitely close to each other. Also, $TF_{m2}$ and $TF_2$ are infinitely close to each other. We can confirm that the obtained parameter matrices $A_{mi}, B_{mi}, C_{mi}, D_{mi}, i = 0,\cdots, 5$ are well approximated with the parameter matrix of the given plant. Therefore, we can conclude from these results that the objective of the paper, shown in Problem \ref{prob1}, is successfully achieved. 

We also note that the numerical results demonstrate that, unlike conventional methods, our approach does not rely on specific periodic input signals. 

\section{Conclusion}\label{sec6}

In this paper, we propose a cyclic reformulation-based system identification algorithm for multirate systems. First, properties of the cyclic reformulation of multirate systems are derived. The controllability and observability of the cycled system are characterized. In addition, the characteristics of the Markov parameters are summarized. An algorithm for obtaining plant parameters using a given multirate data is proposed. A coordinate transform matrix for the identified model parameters based on cycled signals is proposed based on the characteristics of Markov parameters of the cycling structure. The effectiveness of the proposed system identification algorithm for multirate systems is verified using numerical examples. Especially, it is noteworthy that equations $TF_1 = TF_{m1}$ and $TF_2 = TF_{m2}$ are satisfied by our proposed algorithm. Our approach is anticipated to provide practical solutions to various challenges associated with the modeling of control systems with complex sensor networks.

While our numerical examples consistently validate Assumption \ref{ass1} regarding the sparse structure of the Markov parameters, future work will focus on developing a formal mathematical proof of this property and its stability in the presence of noise. This theoretical foundation will support the creation of more robust identification methods that can effectively handle real-world noise while maintaining high accuracy.

\appendix
\section{Proof of Theorem \ref{theo1}}
From Assumption \ref{ass1}, $\D_*$ is a block diagonal matrix. We aim to prove that $T^{-1} \A_* T$ and $T^{-1} \B_*$ are cyclic matrices, and that $\C_* T$ is a block diagonal matrix.

If Assumption \ref{ass1} holds, it can be shown that the following matrix exhibits a block diagonal structure. 
\begin{eqnarray}
\C_* T =  \sum_{i = 0}^{n-1}\sum_{j = 0}^{M-1} \C_* \A_*^{Mi+j}\B_*  \check G_{(Mi+j) \bmod n}  \check S_n^{j+1}\label{proofeq1}
\end{eqnarray}
This is because $ \C_* \A_*^{Mi+j}\B_* \check S_n^{j+1}$ is a block diagonal matrix, and $\check G_{(Mi+j) \bmod n} $ is also a block diagonal matrix for any $i$ and $j$. Consequently, $\C_*T$ is a block diagonal matrix.

Then, we prove that $T^{-1}\B_*$ is given as a cyclic matrix. 
Following a similar calculation to (\ref{proofeq1}), matrices $\check S_l^j \C_* \A_*^j T$ ($k=0,1,\cdots$) are regarded as block diagonal matrices whose matrices sizes are $Ml\times Mn$. 

%Then, we determine block diagonal matrices $\check G_j (j=0,\cdots,n-1)$ whose matrix sizes are $Mm \times Mn$. Their block components are given as $G_j$, whose sizes are $m\times n$, in a similar manner as $F_j$. As a simple example of $G_j$ for $m=1$, each $G_j$ has a size of $1 \times n$. The components of $G_j$ are given as follows: 
%\begin{eqnarray}
%(G_j)_i = \begin{cases} 1, i = j\\0, i\neq j\end{cases}, 
%\end{eqnarray}
%ここで，Fの定義を行う．また，Fに関する式を出してくる．

Since the matrices $\check S_l^{j} \C_* \A_*^j T$ are block diagonal matrices for any $j$, the following matrix $X_*$ can be regarded as a block diagonal matrix whose size is $Mn\times Mn$. 
\begin{eqnarray}
X_* = \sum_{i=0}^{n-1}\sum_{j = 0}^{M-1} \check F_j \check S_l^{j} \C_* \A_*^{Mi+j} T. \label{proofeq5}
\end{eqnarray}

We also confirm that $\check S_l^{k}\C_*\A_*^k\B_* $ ($k=0,1,\cdots$) are given as cyclic matrices based on Assumption \ref{ass1}. Therefore, $X_* T^{-1} \B_*$ is also given as a cyclic matrix. In this paper, we assume the multirate system to be observable, and $X_*$ is considered a regular matrix through the appropriate choice of $\check F_j$. Since $X_*$ is a block diagonal matrix, $X_*$ is invertible, and $X_*^{-1}$ is also a block diagonal matrix. By multiplying the block diagonal matrix $X_*^{-1}$ from the right-hand side of the cyclic matrix $X_* T^{-1}\B_*$, it is clear that $T^{-1} \B_*$ can be represented as a cyclic matrix.

%%%%%%%%%%%%%%%%%%%%%%%%%%%%%%%%%%%%%%%%%%%%%%%%%%
%上記までの証明は終わった。
%(以下，サイクリックとするかブロック対角とするかが正しいか、YでなくXを使う証明になるか慎重に検討する必要がある．) ⇒ 済

Finally, we prove that $T^{-1}\A_* T$ is given as a cyclic matrix. 
The following matrix $Z_{ij}$ is given as a cyclic matrix for any $i,j$ from Assumption \ref{ass1}.
\begin{eqnarray}
Z_{ij} = \check S_l^{i}\C_*\A_*^{i+j} \B_* \check S^{j}
\end{eqnarray}
In addition, $Z_{ij}$ can be rewritten as follows. 
\begin{eqnarray}\label{proofeq2}
Z_{ij} = \check S_l^{i-1}\C_*\A_*^{i-1}T\, T^{-1}\A_* T\, T^{-1} \A_*^{j} \B_* \check S_m^{j+1}
\end{eqnarray}
By using the block diagonal matrices $\check F_i$ and $\check G_j$, it is obvious that the following terms are given as cyclic matrices for any $i,j$. 
\begin{eqnarray}\label{proofeq3}
\check F_{i-1} Z_{ij} \check G_j
\end{eqnarray}
Then, the following matrix $Z$ is obviously given as a cyclic matrix by using the characteristics about (\ref{proofeq2}) and (\ref{proofeq3}).
\begin{eqnarray}\label{proofeq4}
Z &= \left(\sum_{i=1}^{Mn} \check F_{i-1}\check S_l^{i-1}\C_*\A_*^{i-1}T\right) T^{-1}\A_* T\nonumber \\&\cdot T^{-1}\left(\sum_{j=0}^{Mn-1} \A_*^{j} \B_* \check S_m^{j+1}\check G_j\right)
\end{eqnarray}
In (\ref{proofeq4}), the right hand side term $\left( T^{-1} \sum_{j=0}^{Mn-1}\A_*^{j} \B_* \check S_m^{j+1}\check G_j \right)$ is an identity matrix by (\ref{checky3}). The right-hand side term is $X_*$. Therefore, the following equation holds.
\begin{eqnarray}\label{proofeq6}
Z = X_* T^{-1}\A_* T 
\end{eqnarray}
By multiplying the block diagonal matrix $X_*^{-1}$ from right hand side in (\ref{proofeq6}), $T^{-1} \A_* T$ is given by $T^{-1}A_* T = X_*^{-1} Z$. Since $Z$ is a cyclic matrix and $X_*^{-1}$ is a block diagonal matrix, $T^{-1}A_* T$ is a cyclic matrix. This concludes the proof. 

%\begin{table}[t]
%	\begin{center}
%		\caption{Caption of table.}
%		\label{table}
%		\renewcommand{\arraystretch}{1.5}
%		\begin{tabular}{|c|c|c|c|c|}
%			\hline
%			\textbf{Variable} & \textbf{N} & \textbf{Mean}& 
%			\textbf{Max} & \textbf{Min}  \\ \hline
%			A & 10& 50& 100& 0  \\ \hline
%			B & 20& 80& 200& 0  \\ \hline
%		\end{tabular}
%	\end{center}
%\end{table}

%\begin{figure}[t]
%	\centering
%	\includegraphics[width=8cm]{./fig/fig.jpg}
%	\caption{Caption for fig.}
%	\label{fig}
%\end{figure}

%プログラムコードなどを紹介する場合
%\noindent
%\begin{onlinematerials}{99}
%	\bibitem[a]{web}%
%	``JRM Website.''
%	\url{https://www.fujipress.jp/jrm/}
%	[Accessed January 22, 2024]
%\end{onlinematerials}

\end{document}